\documentclass{article}
\usepackage{amsmath,amssymb,amsthm,amsfonts,amscd, graphicx,mathrsfs}
 \usepackage{url} 
\usepackage{epsfig}
\usepackage{subfigure}
\usepackage{tabularx}
\newcommand{\EE}{\mathbb{E}}

\theoremstyle{plain}
\newtheorem{theorem}{Theorem}[section]
\newtheorem*{theorem*}{Theorem}

\newtheorem{proposition}[theorem]{Proposition}
\theoremstyle{definition}
\newtheorem{definition}{Definition}[section]

\theoremstyle{remark}
\newtheorem*{remark}{Remark}

\title{Single Nugget Kriging} 

\author{Minyong R.\ Lee\footnotemark[2] \and 
Art B.\ Owen\footnotemark[2]}

\begin{document}
\maketitle
\newcommand{\slugmaster}{
\slugger{juq}{xxxx}{xx}{x}{x--x}}

\renewcommand{\thefootnote}{\fnsymbol{footnote}}

\footnotetext[2]{Department of Statistics, Stanford University, Stanford, CA 94305. }

\renewcommand{\thefootnote}{\arabic{footnote}}

\begin{abstract}
We propose a method with better predictions at extreme values than the standard method of Kriging. We construct our predictor in two ways: by penalizing the mean squared error through conditional bias and by penalizing the conditional likelihood at the target function value. Our prediction exhibits robustness to the model mismatch in the covariance parameters, a desirable feature for computer simulations with a restricted number of data points. Applications on several functions show that our predictor is robust to the non-Gaussianity of the function. 
\end{abstract}

\pagestyle{myheadings}
\thispagestyle{plain}
\markboth{MINYONG R.\ LEE AND ART B.\ OWEN}{SINGLE NUGGET KRIGING}

\section{Introduction} \label{sec:intro}
In many fields of engineering and science, computer experiments have become an essential tool in studying physical processes such as the subsurface of the earth, aerodynamic forces on bridge decks, and channel network flow. These experiments can be thought of as functions: given a set of input variables in a fixed domain the computer experiment returns the output, which can be a single value, a vector, or even a function. These experiments are usually deterministic, that is if we run the experiment with the same set of input variables, the output is identical. For more discussions of problems and examples in computer experiments, see Sacks et al. \cite{sacks1989design} and Koehler and Owen \cite{koehler1996computer}.

Kriging is a popular way to build metamodels in computer experiments. The method was initially proposed by D.G. Krige \cite{kbiob1951statistical}, and improved by G. Matheron \cite{matheron1963principles}. Kriging exactly interpolates the experimental data and produces predictions at unobserved inputs. The method also generates credible intervals which represent the uncertainty of the prediction. Stein \cite{stein1999interpolation} and Switzer \cite{switzer2006kriging} give summaries and in-depth discussions of Kriging.

However, there are several limitations of Kriging. First of all, the Kriging prediction depends on the covariance hyperparameters that are usually unknown and need to be estimated. The variability of the predicted process highly depends on the hyperparameters, and the likelihood of the hyperparameters are usually computationally expensive to compute and could have many local maxima. There have been several approaches to stabilize the estimation of the hyperparameters, such as Covariance Tapering by Kaufman et al. \cite{kaufman2008covariance} and Penalized Kriging by Li and Sudjianto \cite{li2005analysis}. We would like to find a predictor that is less affected by the hyperparameters.

Secondly, the Kriging prediction depends on the \textit{mean function} that we need to specify before looking at the data. In Kriging, there is a ``regression effect'', in which the predictions are pulled towards the mean function. This comes from minimizing the overall mean squared prediction error, and may give bad predictions at extreme function values. Conditional Bias-Penalized Kriging (CBPK) by Seo \cite{seo2013conditional} suggests minimizing the mean squared error plus the squared conditional bias to improve the performance at the extreme values. Furthermore, if there is a model mismatch, for instance if the mean function is assumed to be zero but actually it is a linear combination of input values, the predictions can be poor. Limit Kriging by Joseph \cite{joseph2006limit} and Blind Kriging by Joseph \cite{joseph2008blind} mitigate this problem. 
 
In this paper, we propose a new prediction method which we call Single Nugget Kriging (SiNK). In section \ref{sec:kriging}, we briefly introduce Kriging. In section \ref{sec:condlik}, we discuss conditioning the likelihood at the target, a fundamental idea of the SiNK. In section \ref{sec:sink}, we define SiNK, and show that it gives smaller mean squared prediction error than usual Kriging when the function value is far from the mean function. In other words, SiNK is robust to misspecifying the mean function or covariance hyperparameters. In section \ref{sec:numexp}, we compare the performance of SiNK to the performance of usual Kriging and Limit Kriging in several numerical experiments.

\section{Kriging} \label{sec:kriging}

 Kriging, or Gaussian Process Regression, treats the deterministic function $f(\mathbf{x})$ as a realization of a one-dimensional random field 
\begin{equation*}
Y(\mathbf{x}) = m(\mathbf{x}) + Z(\mathbf{x})
\end{equation*}
where $\mathbf{x} \in \mathbb{R}^d$, $m(\mathbf{x})$ is a deterministic mean function, and $Z(\mathbf{x})$ is a stationary Gaussian process with mean zero and covariance function $K(\cdot, \cdot)$.

There are three widely used Kriging models based on the mean function. When the mean function is a known function, it is called Simple Kriging, and when the function is an unknown constant $\beta$, it is called Ordinary Kriging. When the mean function is a linear combination of known functions $f_0,\ldots,f_p$ but coefficients $\beta_0,\ldots,\beta_p$ are unknown, namely $m(\mathbf{x}) = \sum_{k=0}^p \beta_k f_k (\mathbf{x})$, it is called Universal Kriging. 

For the covariance function, stationary covariance functions that are tensor products of one-dimensional kernels are popular. Let $C_{\theta}: \mathbb{R} \rightarrow [-1, 1]$ be a covariance kernel with length-scale parameter $\theta$. Let
\begin{equation*}
K(\mathbf{x}, \mathbf{y} )=\sigma^2 C(\mathbf{h}) = \sigma^2 \prod_{j=1}^d C_{\theta_j} (|h_j|) = \sigma^2 \prod_{j=1}^d C_{1} \left(\frac{|h_j|}{\theta_j} \right) 
\end{equation*} 
where $\mathbf{h} = \mathbf{x}- \mathbf{y}$ and $\sigma^2$ and $(\theta_1,\dots,\theta_d)$ are estimated from the data. Mat\'ern covariance kernels \cite{matern1986spatial} are defined as
\[
C_{\nu,\theta}(d) = \frac{(\sqrt{2\nu} \frac{d}{\theta})^{\nu}}{\Gamma(\nu) 2^{\nu-1}}  K_{\nu} \left(\sqrt{2\nu} \frac{d}{\theta} \right)
\]
where $K_\nu (\cdot)$ is the modified Bessel function of the second kind. Mat\'ern covariance kernels are one of the most commonly used kernels in practice because the smoothness of its process, defined in terms of its mean square differentiability, can be parametrized through $\nu$. 

For high dimensional functions, isotropic covariances 
\begin{equation*}
K(\mathbf{x}, \mathbf{y} ) = \sigma^2 C_\theta(\|\mathbf{h}\|) = \sigma^2 C_1\left(\frac{\|\mathbf{h}\|}{\theta}\right)
\end{equation*} 
are often used, where $\|\cdot\|$ is the Euclidean norm. If there is a measurement error or noise in the function, then adding a nugget effect handles the discontinuity in the function, namely
\begin{equation*}
K(\mathbf{x}, \mathbf{y} )=\sigma^2 C(\mathbf{h}) = \sigma^2 \prod_{j=1}^d C_{\theta_j} (|h_j|) + \tau^2 \mathbb{I}_{0}(\mathbf{h})
\end{equation*}  
where $\tau^2 >0$ is a parameter and $\mathbb{I}_{0}$ is the indicator function of the set $\{0 \} \subset \mathbb{R}^d$. 

Throughout the paper, we only consider deterministic computer experiments and we will use the model with a known (or estimated) constant mean $\beta$ for simplicity. The simplification of the mean function to a constant does not affect predictive performance in general; see Sacks et al. \cite{sacks1989design}. We assume that the hyperparameters of the covariance function are known (or estimated from the data), and we will focus on the prediction at a new point $\mathbf{x}_0$.
   
Now suppose we observe $\mathbf{y} = (Y(\mathbf{x}_1),\ldots,Y(\mathbf{x}_n))$, and let $K = (K_{ij})$ be the $n \times n$ covariance matrix of $\mathbf{y}$, $k(\mathbf{x}_0,\mathbf{x}_0)$ be the variance of $Y(\mathbf{x}_0)$, and $\mathbf{k}(\mathbf{x}_0)$ be the covariance vector between $\mathbf{y}$ and $Y(\mathbf{x}_0)$. In a matrix form,
 \begin{align*}
\mathrm{Var} \left[ \begin{pmatrix}
Y(\mathbf{x}_0) \\
\mathbf{y}
\end{pmatrix} \right] =  \begin{pmatrix}
k(\mathbf{x}_0,\mathbf{x}_0) & \mathbf{k}(\mathbf{x}_0)^T \\
\mathbf{k}(\mathbf{x}_0) & K
\end{pmatrix}.
\end{align*}
Let $\mathbf{1}$ be the $n$-length vector of all ones. Then, 
 \begin{align*}
Y(\mathbf{x}_0) \big|~ Y(X) = \mathbf{y} ~ \sim N(m, s^2)
\end{align*}
where
 \begin{align*}
&m = \beta +\mathbf{k}(\mathbf{x}_0)^T K^{-1} (\mathbf{y} - \beta\mathbf{1}), \; \mbox{and}\\
& s^2 = k(\mathbf{x}_0,\mathbf{x}_0) - \mathbf{k}(\mathbf{x}_0)^T K^{-1} \mathbf{k}({\mathbf{x}_0}).
\end{align*}
That is, the conditional distribution of $Y(\mathbf{x}_0)$
 given $\mathbf{y}$ is $N(m, s^2)$.
The Simple Kriging predictor is defined by the conditional mean
 \begin{align*}
\hat{Y}_\mathrm{K}(\mathbf{x}_0) = \EE[Y(\mathbf{x}_0) \big| \mathbf{y}]
= \beta +\mathbf{k}(\mathbf{x}_0)^T K^{-1} (\mathbf{y} - \beta\mathbf{1}).
\end{align*}

The Kriging predictor is also the Best Linear Unbiased Predictor(BLUP) that minimizes the mean squared prediction error (MSPE). Specifically, for Simple Kriging, the linear unbiased predictor $\hat{Y}(\mathbf{x}_0) = \beta + \lambda^T (\mathbf{y} - \beta\mathbf{1})$ that minimizes
 \begin{align*}
\EE[(Y(\mathbf{x}_0)-\hat{Y}(\mathbf{x}_0))^2]
\end{align*}
with respect to $\lambda$ is the Simple Kriging predictor.

\section{Conditional likelihood at the target and conditional bias} \label{sec:condlik}

In this section, we investigate the idea of maximizing the conditional likelihood given the target function value, which is the supporting idea of the SiNK. We also define a class of predictors by generalizing CBPK.

\subsection{Conditional likelihood at the target}

Let's formulate the prediction problem as an estimation problem. Instead of conditioning by the observed function values, we condition by the unknown function value at the target point and compute the likelihood. We easily find that
\begin{subequations}\label{eq:cl}
\begin{align} 
&Y(X) \big|~ Y(\mathbf{x}_0) = y_0  ~ \sim N(\tilde{m}, \tilde{K}), \; \mbox{where} \\
&\tilde{m} = \beta\mathbf{1} + k(\mathbf{x}_0,\mathbf{x}_0)^{-1}(y_0 - \beta)  \mathbf{k}(\mathbf{x}_0) \;\; \mbox{and}   \\
&\tilde{K} =  K - k(\mathbf{x}_0,\mathbf{x}_0)^{-1} \mathbf{k}(\mathbf{x}_0)  \mathbf{k}(\mathbf{x}_0)^T.
\end{align}
\end{subequations}
Now the conditional mean is a vector and the conditional variance is a matrix. The conditional log likelihood is 
\begin{align}\label{eq:condloglik}
l(y_0) &= -\frac{1}{2}(\mathbf{y} -\tilde{m} )^T \tilde{K}^{-1} (\mathbf{y} -\tilde{m} ) + \mbox{constant}.
\end{align}

Note that the maximizer of the conditional likelihood with respect to $y_0$ with penalty $ -(y_0 - \beta)^2/(2k(\mathbf{x}_0,\mathbf{x}_0))$, which is the \emph{maximum a posteriori} estimate of $y_0$ with the prior distribution \mbox{$y_0 \sim N(\beta,k(\mathbf{x}_0,\mathbf{x}_0))$}, is the Simple Kriging predictor. However, the maximizer of the conditional likelihood without penalty (CMLE) is 
\begin{align*}
\hat{Y}_\mathrm{CMLE}(\mathbf{x}_0) = \beta + \frac{k(\mathbf{x}_0,\mathbf{x}_0)}{ \mathbf{k}(\mathbf{x}_0)^T  K^{-1}  \mathbf{k}(\mathbf{x}_0) }   \mathbf{k}(\mathbf{x}_0)^T  K^{-1} (\mathbf{y} - \beta\mathbf{1}).
\end{align*}
The derivation is in the appendix, section \ref{appendix:cmle}. Let us define 
\begin{align*}
\rho = \rho(\mathbf{x}_0) = \sqrt{\frac{ \mathbf{k}(\mathbf{x}_0)^T  K^{-1}  \mathbf{k}(\mathbf{x}_0) }{k(\mathbf{x}_0,\mathbf{x}_0)}}.
\end{align*}
 Then $\rho(\mathbf{x}_0)^2$ is the variance explained by conditioning divided by the marginal variance of $y_0$. The quantity $\rho(\mathbf{x}_0)$ always lies in $[0,1]$, and can be understood as the correlation between the target function value and the data. The CMLE is obtained by inflating the residual term of the Simple Kriging predictor by $1/\rho(\mathbf{x}_0)^2$. 

\subsection{Conditional Bias}
The CMLE is also unbiased in the sense that
$\EE[\hat{Y}_\mathrm{CMLE}(\mathbf{x}_0)] = \beta$. In addition, $\hat{Y}_\mathrm{CMLE}(\mathbf{x}_0)$ is conditionally unbiased, namely
\begin{align*}
\EE[\hat{Y}_\mathrm{CMLE}(\mathbf{x}_0)\big| Y(\mathbf{x}_0)= y_0] = y_0.
\end{align*}
 However, for Simple Kriging, we have 
\begin{align*}
\EE[\hat{Y}_\mathrm{K}(\mathbf{x}_0)\big| Y(\mathbf{x}_0)= y_0] &= \beta +(y_0 - \beta) \frac{\mathbf{k}(\mathbf{x}_0)^T  K^{-1}\mathbf{k}(\mathbf{x}_0)}{k(\mathbf{x}_0,\mathbf{x}_0)} \\
 &= \beta + \rho(\mathbf{x}_0)^2(y_0 - \beta) \neq y_0
\end{align*}
so that $\hat{Y}_\mathrm{K}(\mathbf{x}_0)$ is conditionally biased. We can expect that for a given $y_0$ which is far from the prior mean, the performance of standard Kriging could be worse than the performance of CMLE.

\subsection{Conditional Bias-Penalized Kriging}
Conditional Bias-Penalized Kriging (CBPK) is defined as the linear unbiased predictor $\hat{Y}(\mathbf{x}_0) = \beta + \lambda^T (\mathbf{y} - \beta\mathbf{1})$ that minimizes the MSPE plus a multiple of squared conditional bias (CB) 
 \begin{align} \label{eq:cbpkobj}
\EE[(y_0 - \hat{Y}(\mathbf{x}_0))^2]+ \delta\EE[(y_0 - \EE[\hat{Y}(\mathbf{x}_0)\big| y_0])^2] \;\; (\mbox{for some } \delta \geq 0 )
\end{align}
with respect to $\lambda$.
Seo \cite{seo2013conditional} suggests that we use $\delta=1$, which leads to the predictor
 \begin{align*} 
\hat{Y}_{\mathrm{CBPK}}(\mathbf{x}_0) &= \beta + \frac{2 k(\mathbf{x}_0,\mathbf{x}_0)}{ k(\mathbf{x}_0,\mathbf{x}_0)+\mathbf{k}(\mathbf{x}_0)^T  K^{-1}\mathbf{k}(\mathbf{x}_0)}  \mathbf{k}(\mathbf{x}_0)^T  K^{-1} (\mathbf{y} - \beta \mathbf{1}) \\
&= \beta + \frac{2}{1+\rho(\mathbf{x}_0)^2}  \mathbf{k}(\mathbf{x}_0)^T  K^{-1} (\mathbf{y} - \beta \mathbf{1}).
\end{align*}

We observe that it is again a predictor with an inflated residual term. Different choices of $\delta$ in \eqref{eq:cbpkobj} will lead to different predictors. If $\delta = 0$, \eqref{eq:cbpkobj} is the objective for Simple Kriging, and thus the minimizer $\hat{Y}_{\mathrm{CBPK}}(\mathbf{x}_0) $ is the Simple Kriging predictor. If $\delta \rightarrow \infty$, the minimizing predictor is the CMLE. This matches with the fact that the CMLE is conditionally unbiased. 

The main question when using a CBPK is: which ratio between MSPE and CB should we use? We seek an automatic way to choose $\delta$ instead of simply using $\delta=1$ or applying a cross-validation-style approach. We suggest varying the ratio spatially, in other words, using an appropriate function of $\mathbf{x}_0$ as $\delta$ in the following section. For any nonnegative $\delta$, the generalized CBPK predictor for a constant mean model of the form
 \begin{align*}
\hat{Y}(\mathbf{x}_0) = \beta + w(\mathbf{x}_0)  \mathbf{k}(\mathbf{x}_0)^TK^{-1} (\mathbf{y} - \beta \mathbf{1})
\end{align*}
where $w(\mathbf{x}_0)  \in [1, 1/\rho(\mathbf{x}_0)^2]$. For every nonnegative $\delta$, there is a corresponding $w(\mathbf{x}_0)  \in [1, 1/\rho(\mathbf{x}_0)^2]$. See appendix section \ref{appendix:cbpk} for details.

\section{Single Nugget Kriging} \label{sec:sink}
In this section, we define the Single Nugget Kriging and discuss its properties. 
\subsection{Definition of SiNK} \label{subsec:defsink}
\begin{definition} \label{def:sink}
The Single Nugget Kriging (SiNK) predictor is defined as
\begin{align*}
\hat{Y}_\mathrm{SiNK}(\mathbf{x}_0) &= \beta + \frac{1}{\rho(\mathbf{x}_0)}\mathbf{k}(\mathbf{x}_0)^T K^{-1} (\mathbf{y} - \beta\mathbf{1})\\
&= \beta +\sqrt{\frac{k(\mathbf{x}_0,\mathbf{x}_0)}{\mathbf{k}(\mathbf{x}_0)^T  K^{-1}  \mathbf{k}(\mathbf{x}_0) } }\mathbf{k}(\mathbf{x}_0)^T K^{-1} (\mathbf{y} - \beta\mathbf{1})
\end{align*}
 which is the maximizer of the conditional likelihood given $Y(\mathbf{x}_0)=y_0$ with penalty
\begin{align*}
pen(y_0) =-  \frac{(y_0 - \beta)^2}{2k(\mathbf{x}_0, \mathbf{x}_0)}\frac{\rho(\mathbf{x}_0)}{(1+\rho(\mathbf{x}_0))}   .
\end{align*}
That is, the implicit prior distribution on $y_0$ is $y_0 \sim N(\beta,k(\mathbf{x}_0,\mathbf{x}_0) (1+1/\rho(\mathbf{x}_0)) $.
\end{definition}

SiNK is defined as the \emph{maximum a posteriori} estimator with a prior distribution on $Y(\mathbf{x}_0)$. We inflate the prior variance only at $\mathbf{x}_0$ by the amount of uncertainty measured by $\rho$, to reduce the dependency on the prior. It is equivalent to assuming an independent Gaussian noise only on $Y(\mathbf{x}_0)$, so we call the method Single Nugget Kriging. 

\textit{\begin{remark} \label{rmk:sink}
\normalfont
The SiNK predictor is the CBPK predictor with $\delta = 1/\rho(\mathbf{x}_0)$; it is the linear unbiased predictor $\hat{Y}(\mathbf{x}_0) = \beta + \lambda^T (\mathbf{y} - \beta\mathbf{1})$ where $\lambda$ is the solution of the optimization problem
 \begin{align*}
 \underset{\lambda}{\text{minimize}}\;\; \EE[(y_0 - \hat{Y}(\mathbf{x}_0))^2]+ \frac{1}{\rho(\mathbf{x}_0)}\EE[(y_0 - \EE[\hat{Y}(\mathbf{x}_0)\big| y_0])^2].
\end{align*}
\end{remark}}

Verifications of Definition \ref{def:sink} and Remark \ref{rmk:sink} are in appendix sections \ref{appendix:sink} and \ref{appendix:cbpk} respectively.
As mentioned in section \ref{sec:condlik}, the ratio $\delta$ is now a function of $\mathbf{x}_0$. The conditional bias penalty is larger when we have less information on the target function value. Penalizing by the conditional bias by an appropriate multiple of the conditional bias squared will improve performance at extreme values. The rationale of using $\delta = 1/\rho(\mathbf{x}_0)$ will be discussed in section \ref{subsec:sinkproperty}. 

\subsection{One-point case} \label{subsec:onepoint}
To illustrate the difference among the predictors, we consider the case when there is only one observation. Let $Y_0$ and $Y_1$ be two output values from a function. We observe $Y_1 = y_1$ and want to predict $Y_0$. The model in this case consists of 
\begin{align*}
\EE \left[\begin{pmatrix}Y_0 \\ Y_1 \end{pmatrix} \right] =\begin{pmatrix} \beta \\ \beta \end{pmatrix} \; \mbox{and}\; \mathrm{Var} \left[\begin{pmatrix}Y_0 \\ Y_1 \end{pmatrix}\right]= \sigma^2 \begin{pmatrix}
1 & \rho\\
\rho & 1
\end{pmatrix} 
\end{align*}
where $\rho >0$. The Simple Kriging predictor and the CMLE are
\begin{align}\label{eq:oneptcase}
\hat{Y}_{\mathrm{K}}  &= \beta + \rho (y_1 -\beta) \;\; \mbox{and} \notag \\
\hat{Y}_{\mathrm{CMLE}}  &= \beta + \frac{1}{\rho} (y_1 -\beta).
\end{align}
If we have $\rho$ close to zero, which is the case when we have little information on $Y_0$, then both predictors have problems. The Simple Kriging predictor will depend mostly on the prior mean $\beta$, and the CMLE predictor will have a large variance if the true function value is far from the prior mean. However, the SiNK predictor is
\begin{align*}
\hat{Y}_{\mathrm{SiNK}}  = \beta + \frac{\rho}{\rho} (y_1 - \beta) = y_1
\end{align*}
which does not depend on any parameters. If one wants to rely more on the data than the prior mean $\beta$, SiNK is preferable to Simple Kriging. Intuitively, not only when $n=1$ but also when $n>1$, SiNK will be more robust to the misspecified mean and covariance than usual Kriging. 

\subsection{Properties}\label{subsec:sinkproperty} The main feature of SiNK is its stability which will be represented as boundedness and localness in this section. The natural question that arises may be the uniqueness of a predictor with these properties. Theorem \ref{thm:localness} shows that the SiNK predictor is the unique predictor with both of these properties, in the class of generalized CBPK predictors with MSPE-CB ratio $\delta$ as a function of $\rho(\mathbf{x}_0)$. 

The following proposition shows that if the covariance function is stationary, then the SiNK predictor is bounded. This is not the case for the CMLE because it is unbounded as $\rho(\mathbf{x}_0)$ approaches 0. For instance, in the one-point case \eqref{eq:oneptcase}, $\hat{Y}_{\mathrm{CMLE}}$ diverges as $\rho \rightarrow 0$.

\begin{proposition}[Boundedness]
 \begin{align} \label{eq:bdd}
|\hat{Y}_\mathrm{SiNK}(\mathbf{x}_0)-\beta | \leq \sqrt{k(\mathbf{x}_0,\mathbf{x}_0)}\sqrt{ (\mathbf{y}-\beta\mathbf{1})^T  K^{-1} (\mathbf{y}-\beta\mathbf{1})  }
\end{align}
Thus, if the covariance function is stationary, then 
 \begin{align} \label{eq:bdd2}
\sup_{\mathbf{x}_0 \in \mathbb{R}^d} |\hat{Y}_\mathrm{SiNK}(\mathbf{x}_0)| < \infty.
\end{align}
\end{proposition}

\begin{proof}
By the Cauchy-Schwartz inequality,
 \begin{align*}
|\hat{Y}_\mathrm{SiNK}(\mathbf{x}_0)-\beta | &= \frac{1}{\rho(\mathbf{x}_0)}|\mathbf{k}(\mathbf{x}_0)^T K^{-1}(\mathbf{y}-\beta\mathbf{1})| \\
&\leq \frac{1}{\rho(\mathbf{x}_0)} \sqrt{ \mathbf{k}(\mathbf{x}_0)^T K^{-1}\mathbf{k}(\mathbf{x}_0)}  \sqrt{ (\mathbf{y}-\beta\mathbf{1})^T  K^{-1} (\mathbf{y}-\beta\mathbf{1}) } \\
 &= \sqrt{k(\mathbf{x}_0,\mathbf{x}_0)}\sqrt{ (\mathbf{y}-\beta\mathbf{1})^T  K^{-1} (\mathbf{y}-\beta\mathbf{1})  }
\end{align*}
and equality holds when $K^{-1/2}\mathbf{k}(\mathbf{x}_0)$ and $K^{-1/2}(\mathbf{y}-\beta\mathbf{1})$ are parallel. If the covariance function is stationary, then the right hand side of  \eqref{eq:bdd} does not depend on $\mathbf{x}_0$, thus \eqref{eq:bdd2} holds. \qquad
\end{proof}

 For a predictor with inflated residual of Simple Kriging predictor to be bounded, the maximum amount of inflation is order of $1/\rho(\mathbf{x}_0)$. Roughly speaking, SiNK is the predictor with maximum inflation of the residual term that satisfies boundedness. 

Now let $J_k$ be a set of points that have different distances from observations in $k$'th coordinate, namely
 \begin{align} \label{eq:jk}
J_k := \{\mathbf{x}_0 ~\big| ~|(\mathbf{x}_0 - \mathbf{x}_j)_k| \neq |(\mathbf{x}_0 - \mathbf{x}_l)_k|~ \mbox{for all} ~j\neq l, j,l \in \{1,2,\ldots,n \} \} 
\end{align}
where $k \in \{1,2,\cdots, d\}$. In Proposition \ref{property:localness} and Theorem \ref{thm:localness}, we assume that the new point $\mathbf{x}_0$ is in $J_k$ to break the ties; we remove a measure zero set to simplify the argument. Also, let us define the neighborhood of an observation $\mathbf{x}_j$ for $j \in \{1,2,\ldots,n \}$ as
 \begin{align} \label{eq:bj}
B (\mathbf{x}_j) := \{\mathbf{x}_0 ~\big| ~K(\mathbf{x}_0,\mathbf{x}_j)> K(\mathbf{x}_0, \mathbf{x}_l)~ \forall l \neq j, l \in \{1,2,\ldots,n \}  \} .
\end{align}
 That is, if $\mathbf{x}_0 \in B (\mathbf{x}_j)$, then $\mathbf{x}_j$ is the closest observation to $\mathbf{x}_0$ in terms of covariance.
\begin{proposition}[Localness] \label{property:localness}
Suppose that the covariance function is a tensor product of stationary kernels with length scale parameter $\theta = (\theta_1,\ldots,\theta_d)$. Then 
      \begin{align*}
 \lim_{\theta_k \rightarrow 0} \sup_{\mathbf{x}_0 \in B(\mathbf{x}_j) \cap J_k } |\hat{Y}(\mathbf{x}_0)- Y(\mathbf{x}_j)| = 0
\end{align*}
where $J_k$ and $B (\mathbf{x}_j)$ are sets of points defined in \eqref{eq:jk} and \eqref{eq:bj} respectively.
\end{proposition}

 Proposition \ref{property:localness} shows that as $\theta_k \rightarrow 0$, if $\mathbf{x}_j$ is the closest observation (in $k$'th coordinate) to $\mathbf{x}_0$, then the SiNK predictor $\hat{Y}(\mathbf{x}_0)$ converges to ${Y}(\mathbf{x}_j)$. In the following theorem, we show that the SiNK predictor is the only predictor that satisfies localness, in the class of generalized CBPK predictors. Note that as $\theta_k \rightarrow 0$, the Simple Kriging predictor converges to the prior mean $\beta$. 

\begin{theorem}[Uniqueness] \label{thm:localness}
Consider a conditional biased penalized kriging predictor
 \begin{align*}
\hat{Y}(\mathbf{x}_0) = \beta + w(\mathbf{x}_0)  \mathbf{k}(\mathbf{x}_0)^T K^{-1} (\mathbf{y} - \beta \mathbf{1})
\end{align*}
such that the covariance function is a tensor product of stationary kernels with length scale parameter $\theta = (\theta_1,\ldots,\theta_d)$, and $w(\mathbf{x}_0) \in [1, 1/\rho(\mathbf{x}_0)^2]$ is a continuous function of $\rho(\mathbf{x}_0)$. Suppose that $\mathbf{x}_1,\ldots,\mathbf{x}_n \in J_k$ \eqref{eq:jk}. If there exists a $k \in \{ 1,2,\ldots,d \}$ such that
 \begin{align} \label{eq:localinlimit}
 \lim_{\theta_k \rightarrow 0} \sup_{\mathbf{x}_0 \in B(\mathbf{x}_j) \cap J_k } |\hat{Y}(\mathbf{x}_0)- Y(\mathbf{x}_j)| = 0
 \end{align}
holds where $J_k$ and $B (\mathbf{x}_j)$ are sets of points defined in \eqref{eq:jk} and \eqref{eq:bj} respectively, then $w(\rho(\mathbf{x}_0)) =1/\rho(\mathbf{x}_0)$, i.e. $\hat{Y}(\mathbf{x}_0)$ is the SiNK predictor. 
\end{theorem}

The proof of Proposition \ref{property:localness} and Theorem \ref{thm:localness} is given in the appendix, section \ref{appendix:thm1}. Restricting $w(\mathbf{x}_0)$ to be a function of $\rho(\mathbf{x}_0)$ enables us to guarantee that $w(\mathbf{x}_0) \in [1, 1/\rho(\mathbf{x}_0)^2]$. For example, $w(\mathbf{x}_0) = 1/\rho(\mathbf{x}_0)$ is always in $[1, 1/\rho(\mathbf{x}_0)^2]$. Another example for necessity of this condition is Limit Kriging (Joseph \cite{joseph2006limit}) where the predictor has $w(\mathbf{x}_0) = 1/(\mathbf{k}(\mathbf{x}_0)^T K^{-1} \mathbf{1})$. The Limit Kriging predictor has the localness property, but is not guaranteed to be a CBPK with nonnegative ratio $\delta$, which means we cannot guarantee better performance at extreme values. 

Figure \ref{fig:zakharov} illustrates the property of SiNK and the difference to Ordinary Kriging. The function used in this figure is the 2-dimensional Zakharov function in $[0,1]^2$, which is
 \begin{align} \label{eq:zakharov}
f(\mathbf{x}) = \sum\limits_{i=1}^d x_i^2 + \bigg(\sum\limits_{i=1}^d0.5 ix_i\bigg)^2 + \bigg(\sum\limits_{i=1}^d0.5 ix_i \bigg)^4
\end{align}
where $d=2$, and the input points are 4 midpoints of the edges of a unit square. We fitted Ordinary Kriging and SiNK with an estimated constant mean and tensor product Mat\'ern $5/2$ covariance. For $\theta_1 = \theta_2 = 1$, the predictions are quite similar because $\rho(\mathbf{x}_0) \approx 1$ for all $\mathbf{x}_0 \in [0,1]^2$. However, when $\theta_1 = \theta_2$ are close to zero, we observe significant differences between the two predictions. We also observe the localness property of SiNK. The $\rho(\mathbf{x}_0)$ are close to zero for most of the plotted points, and thus the Ordinary Kriging predictor is close to the estimated constant mean for points far from the observations. The SiNK predictor uses the function value of the observation that is the closest to the target point. 

The localness property of SiNK is also related to the fact that the SiNK prediction at $\mathbf{x}_0$ only depends on the ratios of the correlations with observed function values. For instance, suppose that we predict at another point $\mathbf{x}_0'$ with covariance vector $\mathbf{k}(\mathbf{x}_0') = c  \mathbf{k}(\mathbf{x}_0)$, where $c$ is in $(0,1)$. Then
 \begin{align*}
\hat{Y}_{\mathrm{SiNK}}(\mathbf{x}_0') =\beta+ \frac{\mathbf{k}(\mathbf{x}_0')^TK^{-1}(\mathbf{y}-\beta\mathbf{1})}{\sqrt{\mathbf{k}(\mathbf{x}_0')^T K^{-1}\mathbf{k}(\mathbf{x}_0')}} = \hat{Y}_{\mathrm{SiNK}}(\mathbf{x}_0).
\end{align*}
Thus, the SiNK prediction at $\mathbf{x}_0'$ is the same as the prediction at $\mathbf{x}_0$. However, the Simple Kriging prediction is shrunk to $\beta$ by a factor of $c$. Thus, even if $\mathbf{x}_0'$ is far away from inputs, only the ratios of the correlation determine the SiNK prediction. In other words, SiNK does not automatically converge to the prior mean $\beta$ as $\mathbf{k}(\mathbf{x}_0) \rightarrow 0$, for instance if one of the $\theta_j \rightarrow 0$.

In practice, even though the prediction is theoretically well bounded, dividing by $\rho(\mathbf{x}_0)$ can be numerically unstable when $\rho$ is close to zero. A practical fix is to use
 \begin{align} \label{eq:eps}
\hat{Y}_\mathrm{SiNK, \epsilon}(\mathbf{x}_0)=\beta + \frac{1}{\max(\rho(\mathbf{x}_0), \epsilon)}  \mathbf{k}(\mathbf{x}_0)^T K^{-1} (\mathbf{y} - \beta\mathbf{1})
\end{align}
for a small $\epsilon$. We use $\epsilon = 10^{-3}$ in our numerical work. A larger $\epsilon$ would protect from bad estimators of length-scale parameters that we did not encounter in our numerical experiments.

\begin{figure}
\centering
\subfigure[$\theta = 1$.] {\label{fig:cov1}\includegraphics[width=60mm]{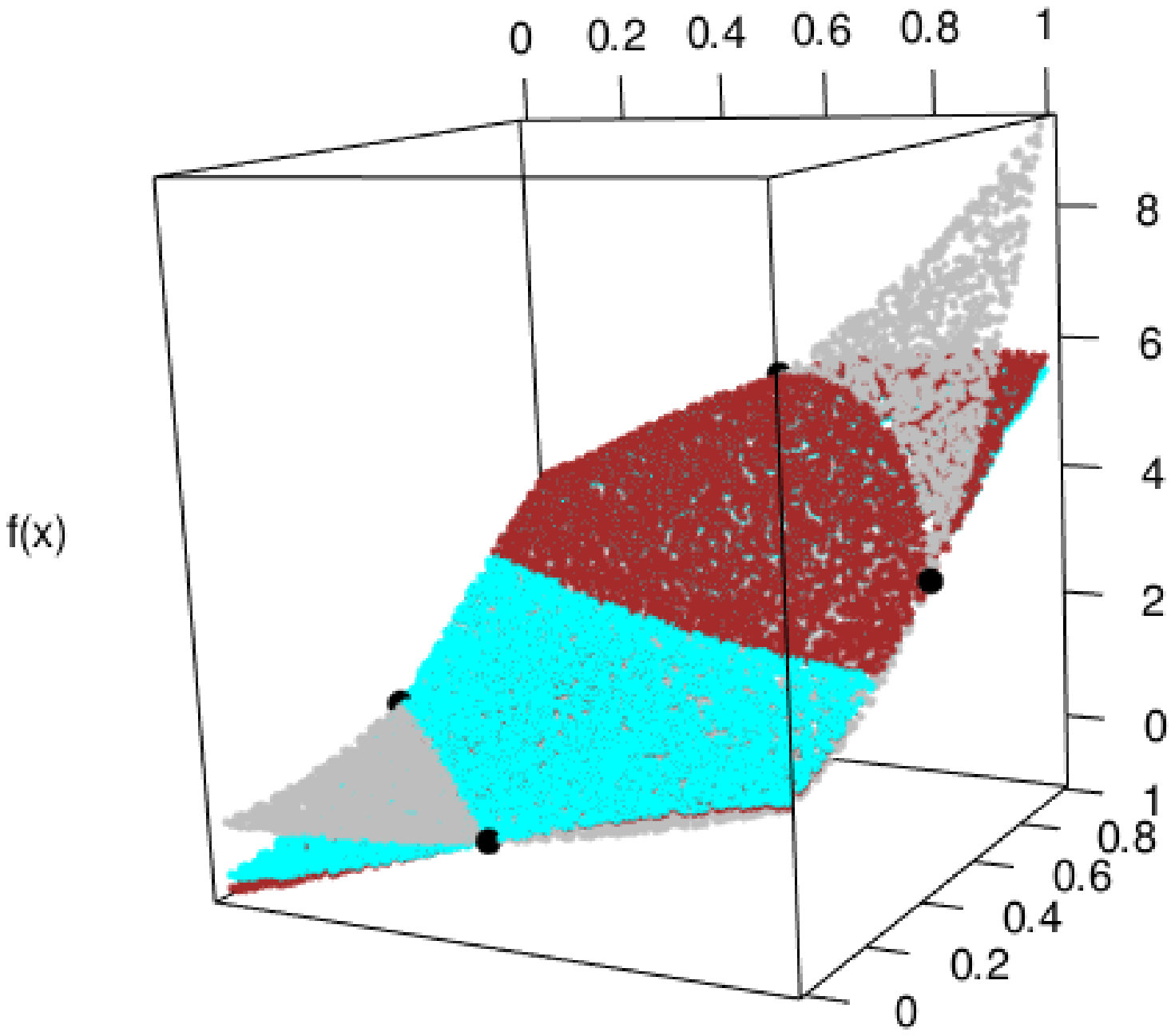}}
\subfigure[$\theta = 0.05$.] {\label{fig:cov005}\includegraphics[width=60mm]{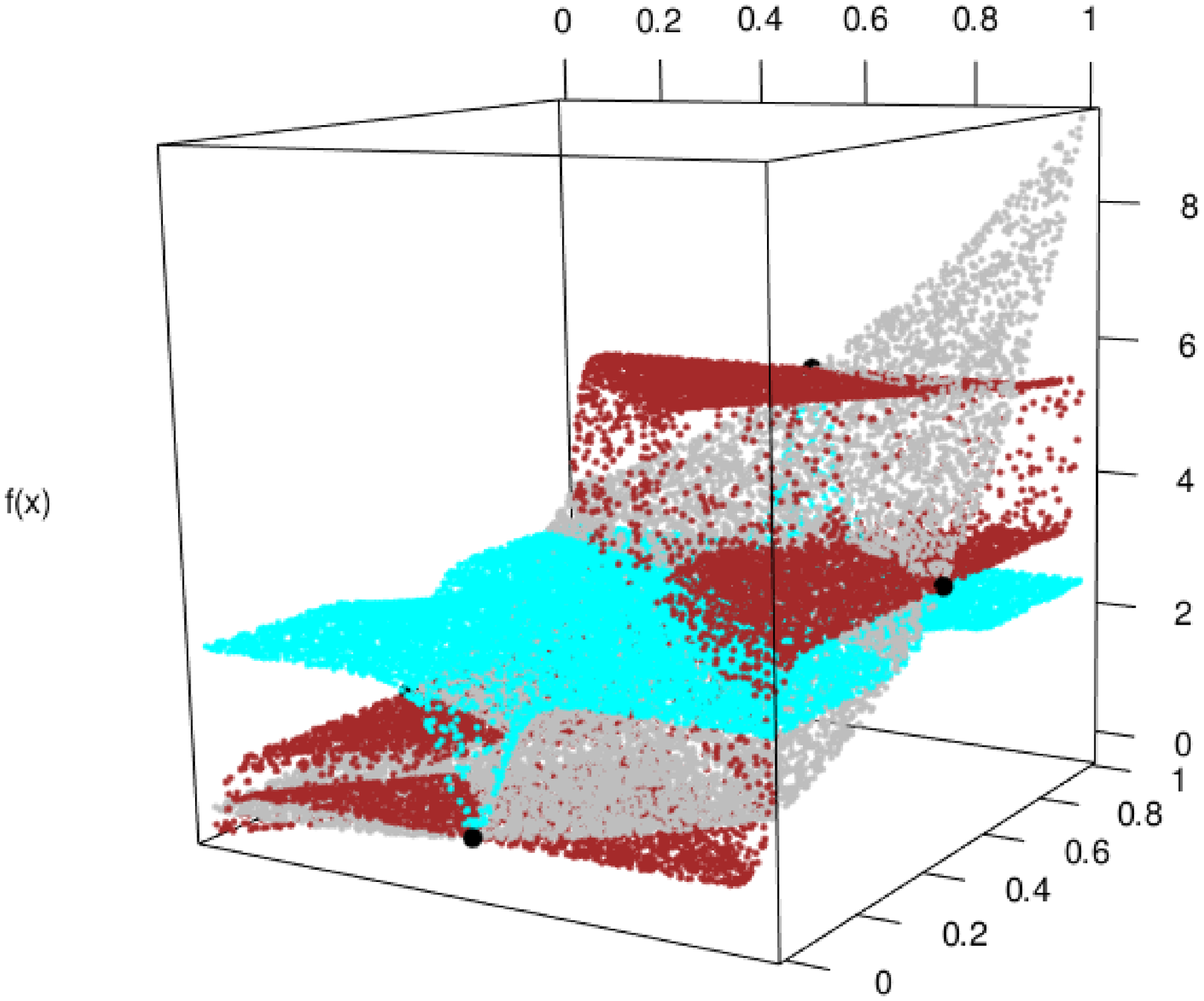}}
\caption{Illustration of the difference between Ordinary Kriging and SiNK. Gray points are the true function values, black points are the observed function values, cyan points are the Ordinary Kriging prediction, brown points are the SiNK prediction. }
\label{fig:zakharov}
\end{figure}

\subsection{Mean squared prediction error at extreme values}\label{subsec:sinkmse}

Since the Simple Kriging predictor is the BLUP, the SiNK predictor has larger MSPE than the Simple Kriging predictor. However, Propositon \ref{prop:mse} tells us that SiNK will be only slightly inferior; the ratio of MSPEs is bounded. 

\begin{proposition} \label{prop:mse}
\begin{align*}
\EE[(\hat{Y}_\mathrm{SiNK}(\mathbf{x}_0)-Y(\mathbf{x}_0))^2] = \frac{2}{1+\rho(\mathbf{x}_0)}  \EE[(\hat{Y}_\mathrm{K}(\mathbf{x}_0)-Y(\mathbf{x}_0))^2]
\end{align*}
That is, the RMSPE of SiNK is at most $\sqrt{2}$ times larger than the RMSPE of Kriging.
\end{proposition}
\begin{proof}
From the conditional distribution of $\mathbf{y}$ given $Y(\mathbf{x}_0) = y_0$ (\eqref{eq:cl}),
\begin{align}\label{eq:cmspe}
& \EE[\,(\hat{Y}_\mathrm{K}(\mathbf{x}_0)-Y(\mathbf{x}_0))^2 \, \big| \, Y(\mathbf{x}_0) = y_0 \,]  \notag \\
& =\mathbf{k}(\mathbf{x}_0)^T  K^{-1}  \mathbf{k}(\mathbf{x}_0) - \frac{(\mathbf{k}(\mathbf{x}_0)^T  K^{-1}  \mathbf{k}(\mathbf{x}_0) )^2}{k(\mathbf{x}_0,\mathbf{x}_0)} + (y_0-\beta)^2 \bigg(1- \frac{\mathbf{k}( \mathbf{x}_0)^T  K^{-1}  \mathbf{k}(\mathbf{x}_0) }{k(\mathbf{x}_0,\mathbf{x}_0)} \bigg) ^2 \; \mbox{and} \notag \\
& \EE[ \, (\hat{Y}_\mathrm{SiNK}(\mathbf{x}_0)-Y(\mathbf{x}_0))^2 \,\big| \, Y(\mathbf{x}_0) = y_0 \,]  \notag \\
& =k(\mathbf{x}_0,\mathbf{x}_0) - \mathbf{k}(\mathbf{x}_0)^T  K^{-1}  \mathbf{k}(\mathbf{x}_0) + (y_0-\beta)^2 \bigg(1- \sqrt{\frac{\mathbf{k}(\mathbf{x}_0)^T  K^{-1}  \mathbf{k}(\mathbf{x}_0) }{k(\mathbf{x}_0,\mathbf{x}_0)}} \bigg)^2
\end{align}
Now since $y_0 \sim N(\beta, k(\mathbf{x}_0, \mathbf{x}_0))$, 
\begin{align*}
\EE[(\hat{Y}_\mathrm{K}(\mathbf{x}_0)-Y(\mathbf{x}_0))^2  ] =k(\mathbf{x}_0,\mathbf{x}_0) - \mathbf{k}(\mathbf{x}_0)^T  K^{-1}  \mathbf{k}(\mathbf{x}_0)
\end{align*}
and finally
\begin{align*}
 \EE[(\hat{Y}_\mathrm{SiNK}(\mathbf{x}_0)-Y(\mathbf{x}_0))^2  ] &= 2k(\mathbf{x}_0,\mathbf{x}_0) - 2\sqrt{k(\mathbf{x}_0,\mathbf{x}_0)  \mathbf{k}(\mathbf{x}_0)^T  K^{-1}  \mathbf{k}(\mathbf{x}_0) }\\
&= \frac{2}{1+\rho(\mathbf{x}_0)}  \EE[(\hat{Y}_\mathrm{K}(\mathbf{x}_0)-Y(\mathbf{x}_0))^2  ] .
\end{align*}
by the definition of $\rho(\mathbf{x}_0)$.\qquad
\end{proof}

Here we show that SiNK has improved performance at extreme values. This can be represented in two ways; conditioning on a single extreme value of $Y(x_0)$ and conditioning on a region of extreme $Y(x_0)$ values. 

\begin{proposition}
If 
\begin{align*}
\bigg|\frac{y_0-\beta}{\sqrt{k(\mathbf{x}_0,\mathbf{x}_0)}} \bigg| \geq \sqrt{\frac{(1+\rho(\mathbf{x}_0))^2}{(1+\rho(\mathbf{x}_0))^2-1}}
\end{align*}
holds, then
      \begin{align*}
\EE[\,(\hat{Y}_\mathrm{SiNK}(\mathbf{x}_0)-Y(\mathbf{x}_0))^2 \, \big| \, Y(\mathbf{x}_0) = y_0 \,] \leq \EE[\,(\hat{Y}_\mathrm{K}(\mathbf{x}_0)-Y(\mathbf{x}_0))^2 \, \big| \, Y(\mathbf{x}_0) = y_0 \,] .
\end{align*}
\end{proposition}
\begin{proof}
Directly follows from \eqref{eq:cmspe}. \qquad
\end{proof}

\begin{proposition}
Let $\phi(\cdot)$ and $\Phi(\cdot)$ be the density function and distribution function of the standard normal distribution respectively. Let  $Z(\mathbf{x}_0) = |(Y(\mathbf{x}_0)-\beta)/(\sqrt{k(\mathbf{x}_0,\mathbf{x}_0)})|$. For $M>0$, if $\rho(\mathbf{x}_0) \geq -1 + \sqrt{1+(1-\Phi(M))/(M\phi(M))}$, then
      \begin{align*}
\EE \left[(\hat{Y}_\mathrm{SiNK}(\mathbf{x}_0)-Y(\mathbf{x}_0))^2 \big| Z(\mathbf{x}_0) \geq M \right]
 \leq \EE \left[(\hat{Y}_\mathrm{K}(\mathbf{x}_0)-Y(\mathbf{x}_0))^2 \big| Z(\mathbf{x}_0) \geq M \right].
\end{align*}
\end{proposition}
\begin{proof}
Let $\rho = \rho(\mathbf{x}_0)$. From \eqref{eq:cmspe}, 
\begin{align*}
&\EE[(\hat{Y}_\mathrm{K}(\mathbf{x}_0)-Y(\mathbf{x}_0))^2 \big| Z(\mathbf{x}_0)] = k(\mathbf{x}_0,\mathbf{x}_0) (\rho^2 - \rho^4 + Z(\mathbf{x}_0)^2(1-\rho^2)^2) \;\; \mbox{and}\\
&\EE[(\hat{Y}_\mathrm{SiNK}(\mathbf{x}_0)-Y(\mathbf{x}_0))^2 \big| Z(\mathbf{x}_0)] =  k(\mathbf{x}_0,\mathbf{x}_0) (1 - \rho^2 + Z(\mathbf{x}_0)^2(1-\rho)^2).
\end{align*}
Using
\begin{align*}
\EE[Z(\mathbf{x}_0)^2 \big| Z(\mathbf{x}_0)>M] = \frac{1}{1-\Phi(M)}\int_{M}^\infty z^2\phi(z) dz = \frac{M\phi(M) + 1-\Phi(M)}{1-\Phi(M)}
\end{align*}
we get the inequality for $\rho \geq -1 + \sqrt{1+(1-\Phi(M))/(M\phi(M))}$. \qquad
\end{proof}

Figure \ref{fig:rhoMall} shows the relation between $\rho(\mathbf{x}_0)$ and the critical $z$-score or the threshold $M$ for $z$-score. The ratio of the region-conditional mean squared prediction error 
\begin{align} \label{eq:cmspedef}
\frac{\mathrm{CMSPE}_{\mathrm{SiNK}}}{\mathrm{CMSPE}_{\mathrm{K}}}=\frac{\EE\left[(\hat{Y}_\mathrm{SiNK}(\mathbf{x}_0)-Y(\mathbf{x}_0))^2\big|Z(\mathbf{x}_0)\geq M \right]}{\EE\left[(\hat{Y}_\mathrm{K}(\mathbf{x}_0)-Y(\mathbf{x}_0))^2\big|Z(\mathbf{x}_0)\geq M \right]}
\end{align}
 decreases as the threshold $M$ increases. 

\begin{figure}
\centering
\subfigure[The level curve for the mean squared error. SiNK outperforms Simple Kriging in the region above and to the right of the given curves.] {\label{fig:rhoM}\includegraphics[width=60mm]{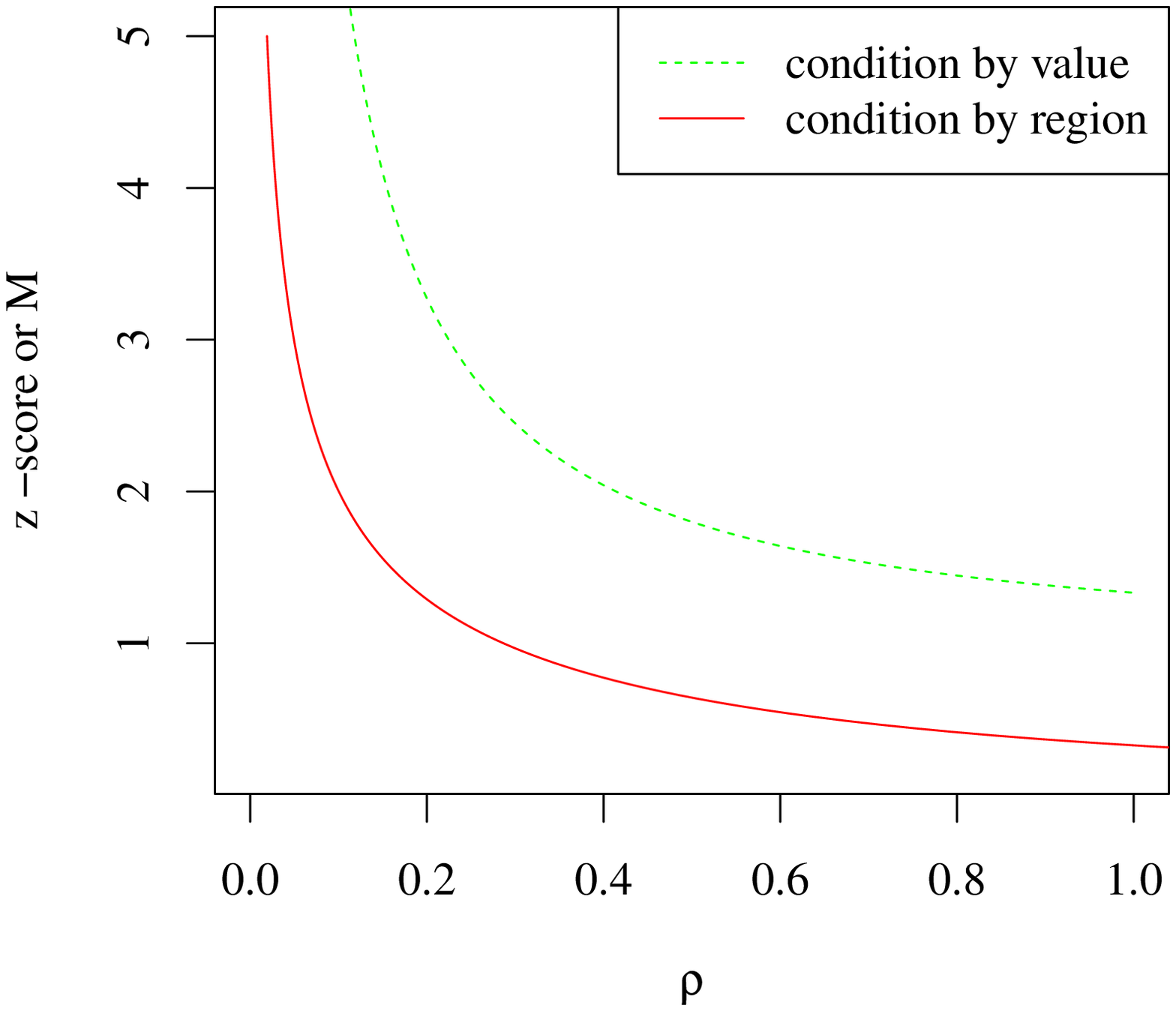}}
\subfigure[Contour plot of $\mathrm{CMSPE}_{\mathrm{SiNK}}/\mathrm{CMSPE}_{\mathrm{K}}$ (\ref{eq:cmspedef}). The ratio decreases as the threshold $M$ increases. ] {\label{fig:rhoMcontour}\includegraphics[width=60mm]{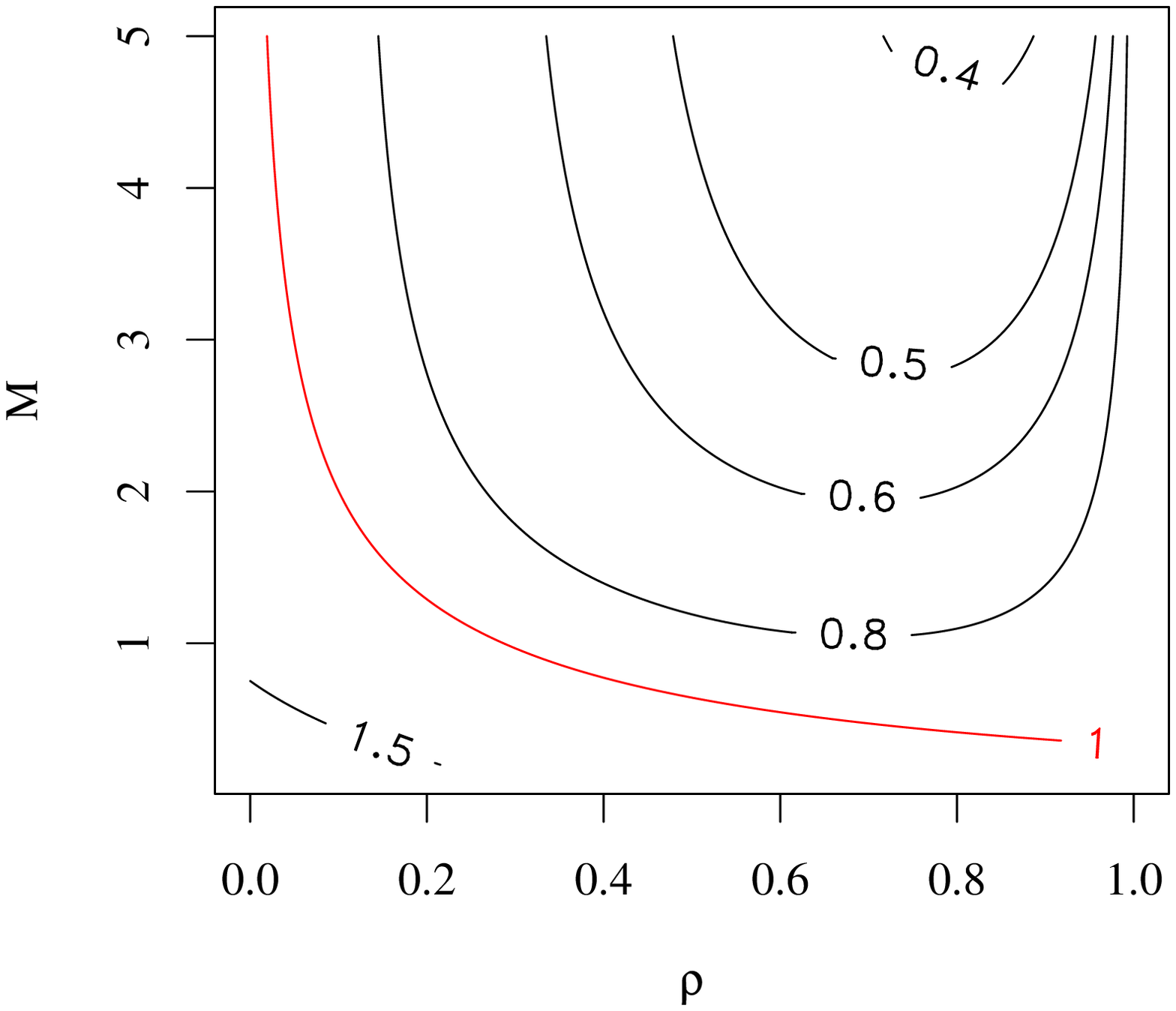}}
\caption{Relation between $\rho(\mathbf{x}_0)$ and $z$-score.}
\label{fig:rhoMall}
\end{figure}

\section{Numerical experiments} \label{sec:numexp}

For numerical simulations, we used the \textsf{DiceKriging} package in \textsf{R} by O. Roustant et al. \cite{roustant2012dicekriging}. We fit the constant mean model for Ordinary Kriging and SiNK, with the maximum likelihood estimator of the constant mean $\hat{\beta} = (\mathbf{1}^TK^{-1}\mathbf{1})^{-1}\mathbf{1}^T K^{-1}\mathbf{y}$. For the covariance function, we used tensor products of Mat\'ern $\nu=5/2$ kernels with maximum likelihood estimators of the length-scale parameters $\theta_1,\ldots,\theta_d$, unless specified otherwise. We used $\epsilon =10^{-3}$ in equation \eqref{eq:eps}.

To measure the performance of a predictor, we computed the empirical integrated squared error (EISE) 
\begin{align*}
\frac{1}{n_T}\sum\limits_{j=1}^{n_T} (\hat{Y}(\mathbf{x}_{test,j}) - Y(\mathbf{x}_{test,j}))^2.
\end{align*}
with an independent set of $n_T$ test points. 
\subsection{Gaussian process}
We generated a realization of a 7 dimensional Gaussian process with zero mean and Mat\'ern covariance with length-scale hyperparameters $\theta = (1,1,1,1,1,1,1)$ and stationary variance $k(\mathbf{x}, \mathbf{x})=\sigma^2 = 1$. The observations were 100 points i.i.d.\ uniform in $[0,1]^7$ and the test points were 2000 points i.i.d.\ uniform in $[0,1]^7$. 

To emulate the real world situation where the hyperparameters are unknown, we estimated the hyperparameters by maximizing the likelihood. The estimated mean was $\hat{\beta} = 0.143$, the estimated length-scale hyperparameters were $\hat\theta = (1.29, 0.92, 1.18, 1.41, 0.95, 0.76, 1.32) $, and the estimated stationary variance was $\hat\sigma^2 = 0.94$. The performance comparison between SiNK and Ordinary Kriging is in Table \ref{table:gppiston}. We observe that SiNK had slightly inferior EISE, but showed better performance at extreme values. 

\begin{center}
\begin{table}[h]
\footnotesize
\centering
\caption{Performance comparison of Ordinary Kriging and SiNK for a realization of Gaussian process and piston function. Extreme values are the function values with $|z$-score$|$ $>$ 2.}
\label{table:gppiston}
\begin{tabular}{c|cc}
\hline
Function                                             & Gaussian Process & Piston Function \\ \hline
Number of observations                               & 100              & 14              \\
$R^2$ Ordinary Kriging                                  & 0.818            & 0.674           \\
$R^2$ SiNK                                              & 0.814            & 0.711           \\
Overall EISE Ratio (SiNK/Ordinary)                   & 1.020            & 0.887           \\
Extreme values EISE Ratio (SiNK/Ordinary) & 0.820            & 0.814           \\ \hline
\end{tabular}
\end{table}
\end{center}

Figure \ref{fig:gp} shows the prediction at test points with extreme function values. We first sort the test points by the true function values and see the 1\% largest and smallest function values. We observe that SiNK reduces the conditional bias by inflating the residual term. Differences are small but consistently in the right direction. 

\begin{figure}
\centering
\subfigure[Prediction at test points with 1\% largest function values.]
 {\label{fig:gp_extremes_large}\includegraphics[width=60mm]{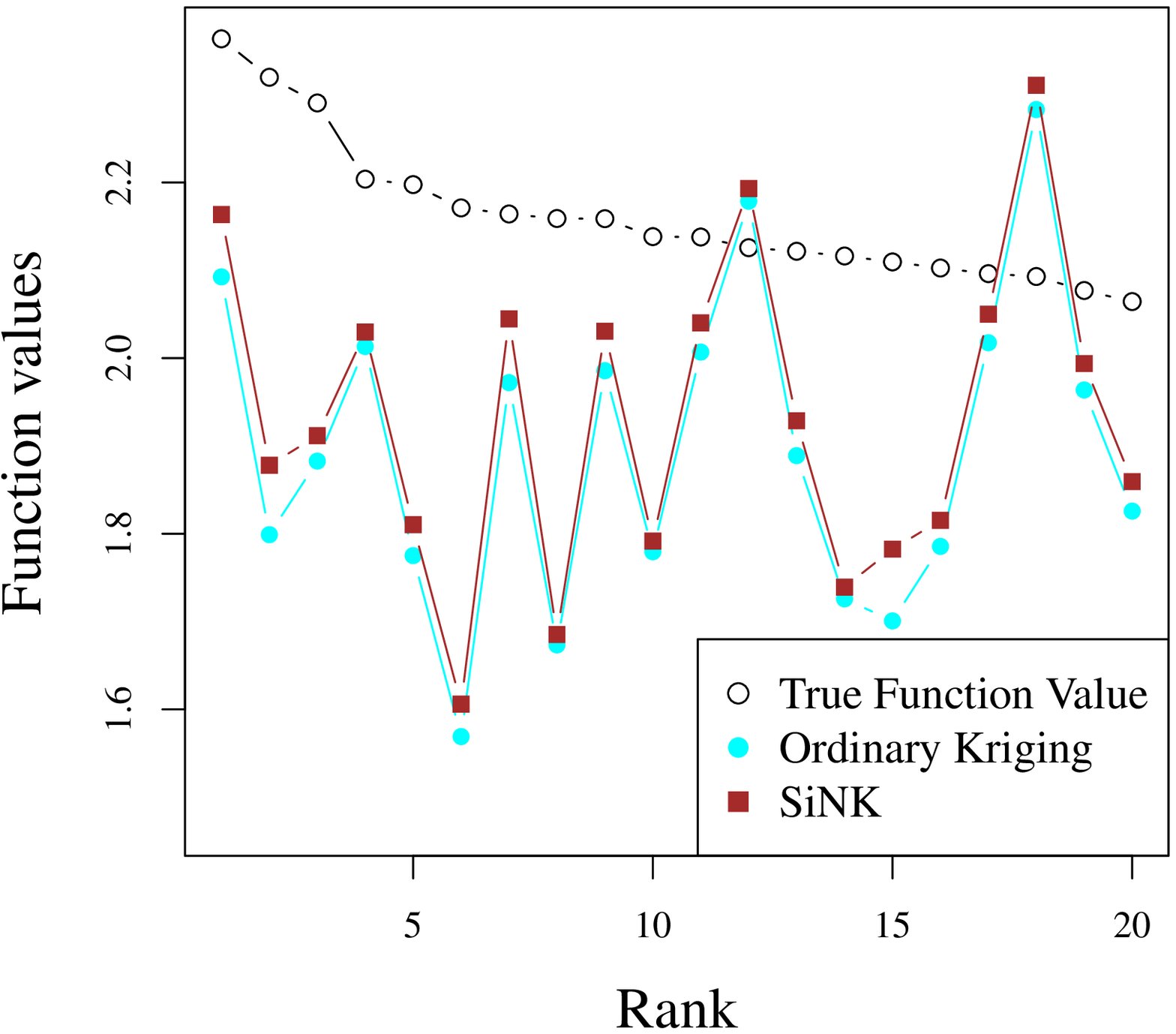}}
\subfigure[Prediction at test points with 1\% smallest function values.]
 {\label{fig:gp_extremes_small}\includegraphics[width=60mm]{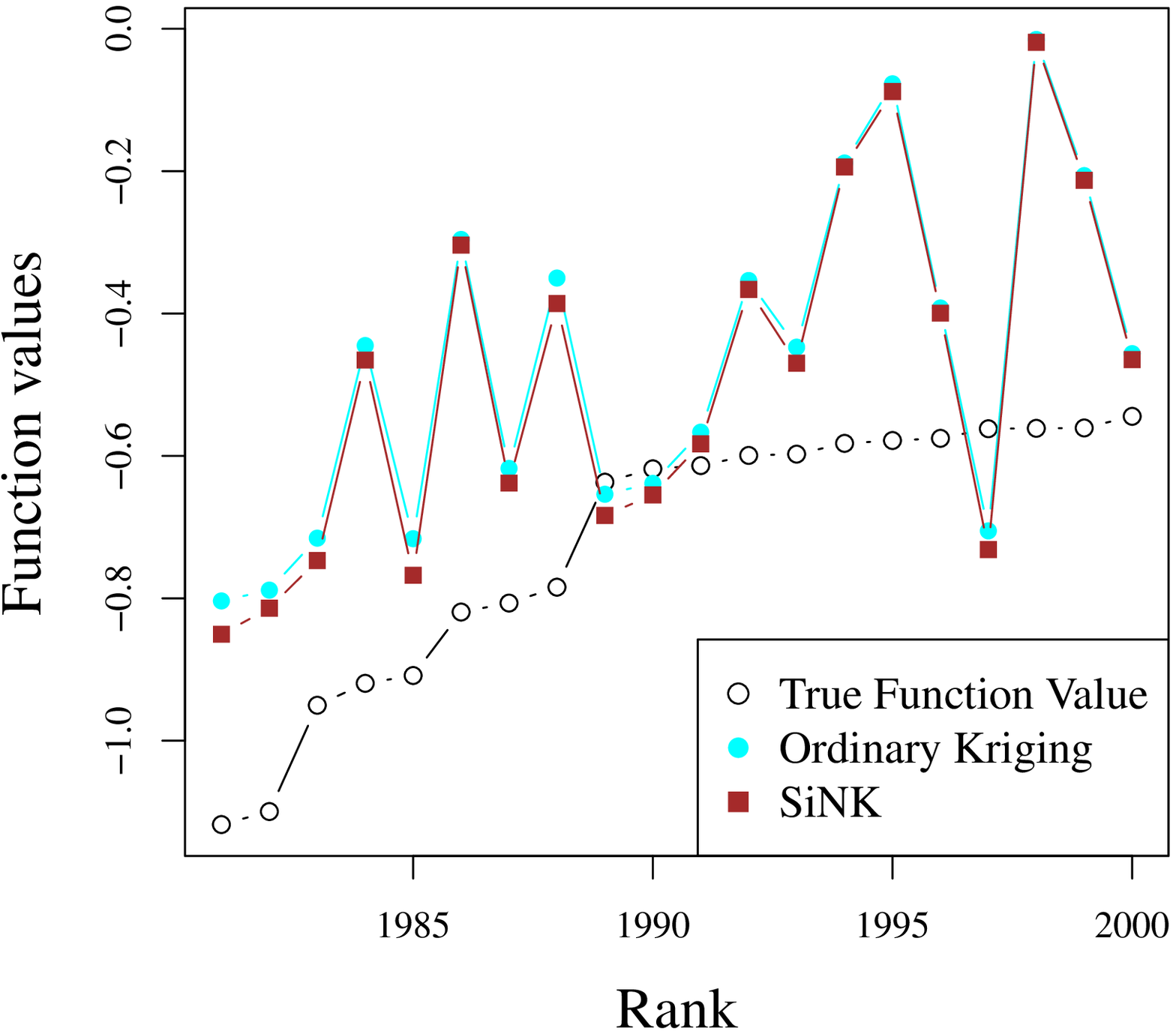}}
\caption{Ordinary Kriging and SiNK for a realization of 7-dimensional Gaussian process. Rank is the order of the true function values of the test points.}
\label{fig:gp}
\end{figure}

\subsection{Piston function}
We examined the performance of SiNK in a computer experiment; the piston simulation function. The piston simulation function in Zacks \cite{zacks1998modern} models the circular motion of a piston within a cylinder. The response $C$ is the time it takes to complete one cycle, in seconds. The formula of the function is
\[
C(\mathbf{x}) = 2\pi \sqrt { \frac{M}{k+ S^2 \frac{P_0 V_0 }{T_0} \frac{T_a}{V^2}}}
\]
where
\begin{align*}
V= \frac{S}{2k} \left( \sqrt{A^2 + 4k \frac{P_0 V_0}{T_0} T_a} - A \right) \;\mbox{and} ~ A = P_0 S + 19.62 M - \frac{kV_0}{S}.
\end{align*}
The description of the input variables is in Table \ref{table:piston}.
\begin{center}
\begin{table}[h]
\footnotesize
\centering
\caption{Input variables $\mathbf{x}$ for the piston function.}
\label{table:piston}
\begin{tabular}{l|l}
\hline
$M \in [30,60]$           & piston weight (kg)             \\
$S \in [0.005, 0.020]$    & piston surface area ($m^2$)    \\
$V_0 \in [0.002, 0.010]$  & initial gas volume ($m^3$)     \\
$k \in [1000, 5000]$      & spring coefficient ($N/m$)     \\
$P_0 \in [90000, 110000]$ & atmospheric pressure ($N/m^2$) \\
$T_a \in [290, 296]$      & ambient temperature (K)        \\
$T_0 \in [340, 360]$      & filling gas temperature (K)    \\ \hline
\end{tabular}
\end{table}
\end{center}

In computer experiments, the design of inputs is also very important, because each experiment is expensive, and a clever design could reduce the approximation error.  Here we adopted Randomized QMC design (Faure sequence base 7) for observations and test points. In Table \ref{table:gppiston}, we see that in this case SiNK performs better not only at extreme values but also overall. This result possibly comes from non-Gaussianity of piston function; more specifically, the reduction of conditional bias may have had a large effect in the test error in this case.

Again, in Figure \ref{fig:piston} the SiNK predictions are better at the test points with extreme function values than the Ordinary Kriging predictions, and the difference is significant at the test points with 1\% smallest function values. The inflation of the residual is consistently in the right direction, and larger than that of the Gaussian process example.

\begin{figure}
\centering
\subfigure[Prediction at the test points with 1\% largest function values.]
 {\label{fig:piston_extremes_large}\includegraphics[width=60mm]{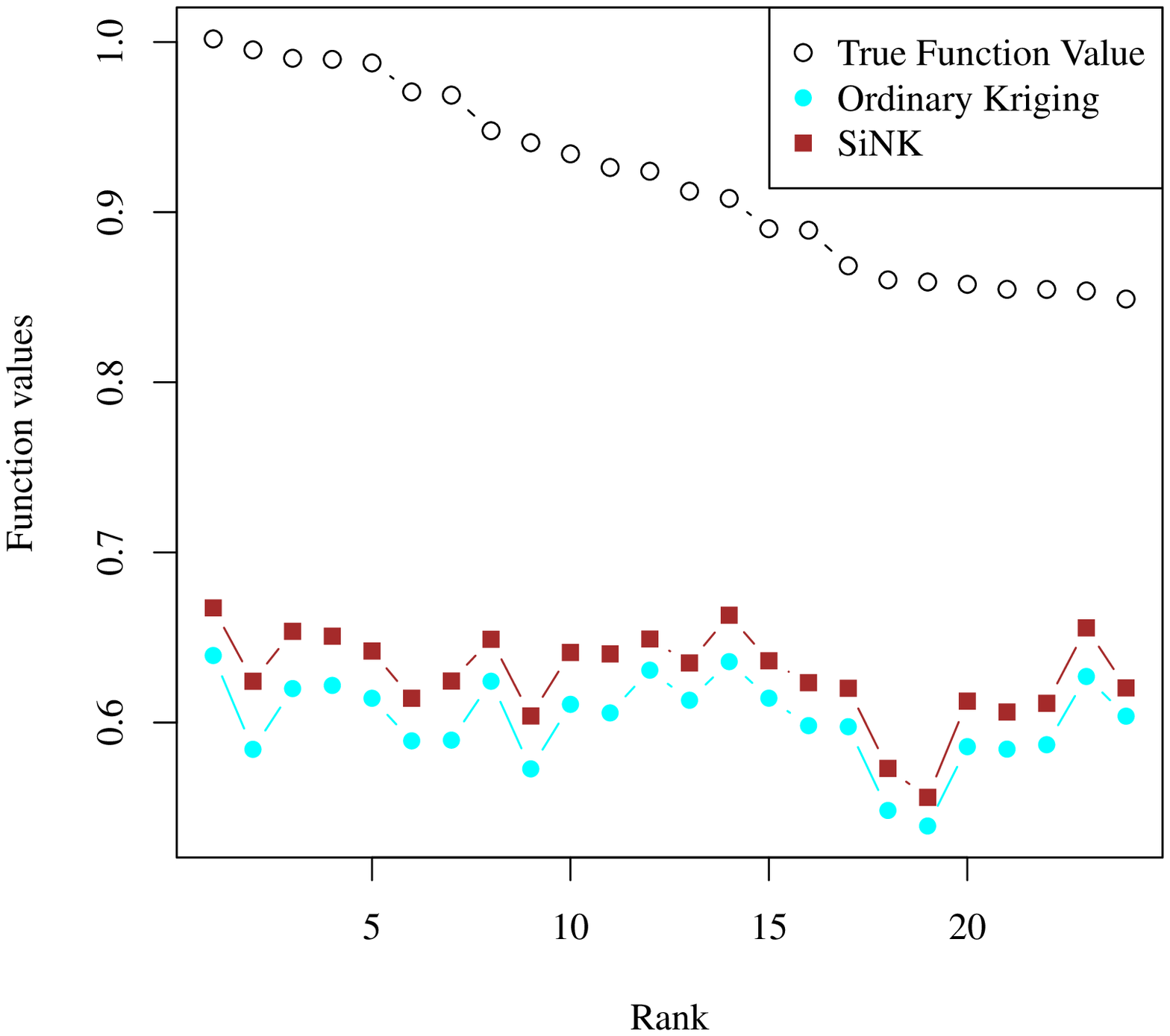}}
\subfigure[Prediction at the test points with 1\% smallest function values.]
 {\label{fig:piston_extremes_small}\includegraphics[width=60mm]{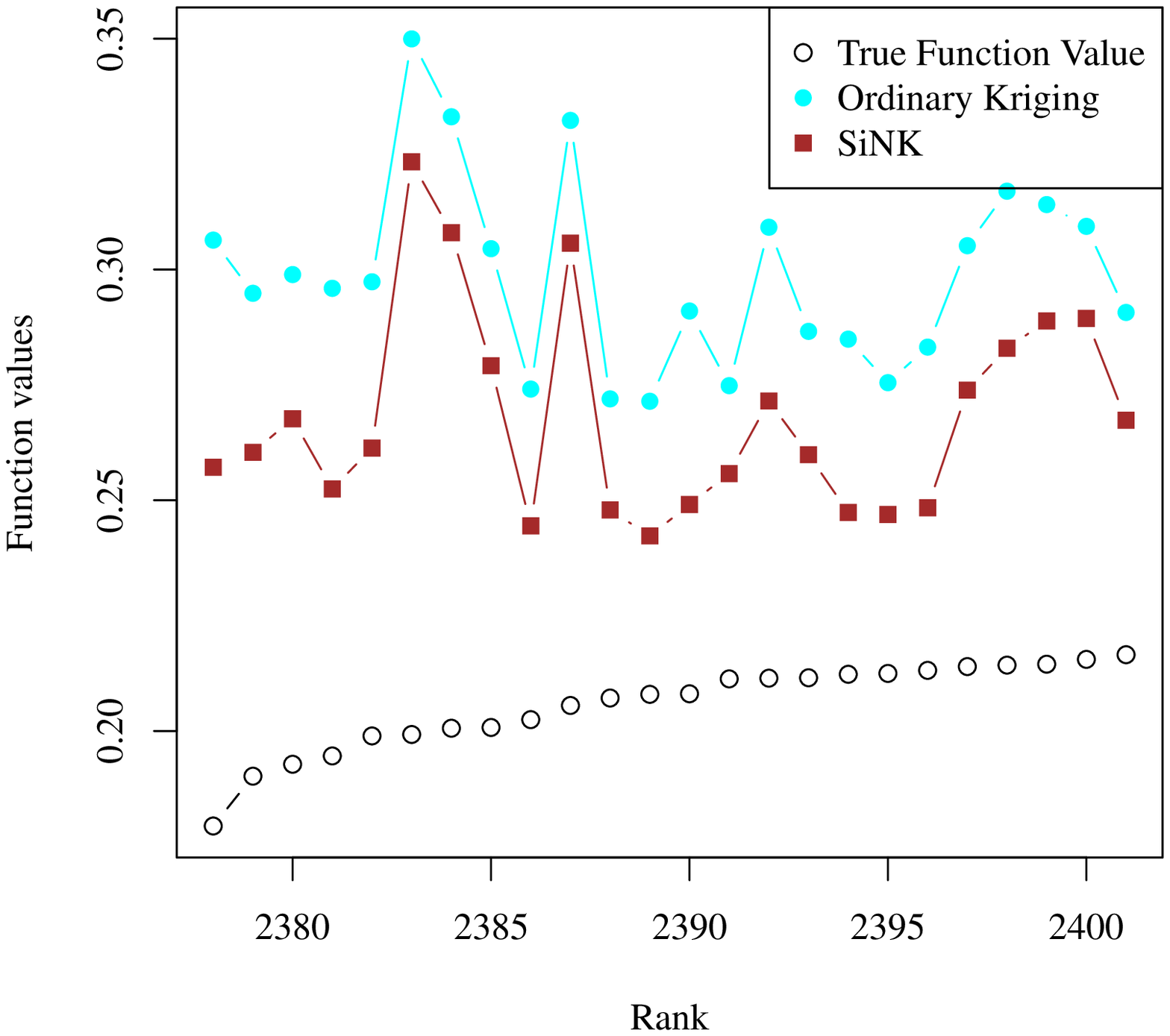}}
\caption{Ordinary Kriging and SiNK for the piston function. Rank is the order of the true function values of the test points.}
\label{fig:piston}
\end{figure}

\subsection{Other functions}
We fit Ordinary Kriging, Limit Kriging and SiNK for several deterministic functions and compared the performances. The test function codes are from Bingham's website (Bingham \cite{bingham2015}).
Table \ref{table:othersimulations} shows the dimension of the function, the number of observed points and test points, covariance type, $R^2$, overall EISE ratio, and EISE ratio at extreme values for each function. The training points and test points are independent and uniformly distributed in the domain of inputs. The number of training points for fitting each function was chosen so that the $R^2$ of Ordinary Kriging is roughly $0.95$, except for fitting the Robot Arm function which is a comparably difficult function to fit with our prediction methods. 

We see that for the 5 functions that we consider, SiNK performed better than Ordinary Kriging in terms of EISE, and the EISE ratios are even smaller for extreme values. Small $R^2$ gains are relevant because large $1-R^2$ improvements are captured by the EISE ratios. For instance, for the Welch function, the SiNK predictions at points with extreme function values (function values such that $|z\mbox{-score}|>2$) have roughly half EISE of the EISE of Ordinary Kriging predictions. In addition, we observe that the performance of Limit Kriging and SiNK is very similar in terms of overall EISE. Limit Kriging also shows improved performance at extreme values compared to Ordinary Kriging, but the improvement is smaller or no different than the improvement of SiNK. For the Friedman function, there was not a test point function value which had $|z$-score$|$ larger than 2. This was due to the large estimate of the stationary variance $\sigma^2 = k(\mathbf{x}, \mathbf{x})$. A suspicious estimate of the stationary variance can be found occasionally in practice, but it is not a problem for the prediction because all three predictors that we are comparing do not depend on the estimate of $\sigma^2$. 
 See appendix section \ref{subsec:testfunctions} for details of the functions used in Table \ref{table:othersimulations}. 

\begin{table}
\noindent\makebox[\textwidth]
{%
\begin{tabularx}{1.11\textwidth}{XX}
\caption{Performance comparison among Ordinary Kriging, Limit Kriging and SiNK. Mat\'ern covariance with $\nu=5/2$ and estimated length-scale parameters are used. \textnormal{NaN} is the case where no function values had $|z\mbox{-score}|>2$.}
\label{table:othersimulations}
\begin{scriptsize}
\setlength{\tabcolsep}{0.13cm}
\begin{tabular}{c|ccccc}
\hline
Function                                                        & Borehole     & Welch        & Piston       & Friedman     & Robot Arm    \\ \hline
Dimension                                                       & 8            & 20           & 7            & 5            & 8            \\
Number of training, test points                                 & 32, 5000     & 320, 5000    & 49, 5000     & 50, 5000     & 512, 5000    \\
$R^2$ (Ordinary Kriging)                                        & 0.934        & 0.948        & 0.962        & 0.967        & 0.854        \\
$R^2$  (Limit Kriging)                                          & 0.942        & 0.961        & 0.968        & 0.968        & 0.858        \\
$R^2$  (SiNK)                                                   & 0.946        & 0.961        & 0.967        & 0.968        & 0.855        \\
Overall EISE Ratio (Limit/Ordinary)                             & 0.884        & 0.744        & 0.843        & 0.977        & 0.970        \\
Overall EISE Ratio (SiNK/Ordinary)                              & 0.819        & 0.750        & 0.876        & 0.991        & 0.992        \\
\multicolumn{1}{l|}{Extreme values EISE Ratio (Limit/Ordinary)} & 0.876        & 0.630        & 0.828        & NaN          & 0.866        \\
Extreme values EISE Ratio (SiNK/Ordinary)                       & 0.803        & 0.489        & 0.834        & NaN          & 0.681        \\ \hline
\end{tabular}
\end{scriptsize}
\end{tabularx}}
\end{table}

\section{Discussion}
We have presented an alternative to Kriging with improved predictions at the extreme values. We first found a link between conditional likelihood at the target and CBPK, and used it to define SiNK. In addition, we showed that SiNK has a boundedness and a localness property. In numerical experiments, we observed that SiNK generally performs better not only at extreme values but also in terms of overall integrated squared error. This result is possibly due to the non-Gaussianity of the functions used in the examples.

\section*{Acknowledgements}
This work was supported by NSF grants DMS-1407397 and DMS-1521145.
\bibliographystyle{plain}
\bibliography{mybib}

\section*{Appendix} 
\appendix
\section{Derivation of the CMLE} \label{appendix:cmle}
For simplicity, let $\rho = \rho(\mathbf{x}_0) = \sqrt{\frac{ \mathbf{k}(\mathbf{x}_0)^T  K^{-1}  \mathbf{k}(\mathbf{x}_0) }{k(\mathbf{x}_0,\mathbf{x}_0)}}$. By the Woodbury formula,
\begin{align*}
\tilde{K}^{-1} &=  (K - \mathbf{k}(\mathbf{x}_0)  k(\mathbf{x}_0,\mathbf{x}_0)^{-1}  \mathbf{k}(\mathbf{x}_0)^T)^{-1} \\
&= K^{-1} +  \frac{K^{-1}  \mathbf{k}(\mathbf{x}_0)\mathbf{k}(\mathbf{x}_0)^TK^{-1}}{k(\mathbf{x}_0,\mathbf{x}_0) -\mathbf{k}(\mathbf{x}_0) ^T K^{-1} \mathbf{k}(\mathbf{x}_0)}.
\end{align*}
Therefore
\begin{align} \label{eq:woodbury}
\mathbf{k}(\mathbf{x}_0)^T\tilde{K}^{-1} = \mathbf{k}(\mathbf{x}_0)^TK^{-1} +  \frac{\mathbf{k}(\mathbf{x}_0)^TK^{-1}  \mathbf{k}(\mathbf{x}_0)\mathbf{k}(\mathbf{x}_0)^TK^{-1}}{k(\mathbf{x}_0,\mathbf{x}_0) -\mathbf{k}(\mathbf{x}_0) ^T K^{-1} \mathbf{k}(\mathbf{x}_0)} = \frac{1}{1-\rho^2} \mathbf{k}(\mathbf{x}_0)^TK^{-1}.
\end{align}
Thus, differentiating the conditional log likelihood \eqref{eq:condloglik} with respect to $y_0$, 
\begin{align*}
\frac{\partial  l(y_0)}{\partial y_0} &= \frac{1}{k(\mathbf{x}_0,\mathbf{x}_0)}(\mathbf{y} -\tilde{m} )^T \tilde{K}^{-1} \mathbf{k}(\mathbf{x}_0) =  \frac{1}{(1-\rho^2) k(\mathbf{x}_0,\mathbf{x}_0)}(\mathbf{y} -\tilde{m} )^T K^{-1} \mathbf{k}(\mathbf{x}_0) 
\end{align*}
from \eqref{eq:woodbury}. Solving $\partial  l(y_0)/\partial y_0 = 0$ leads to 
\begin{align*}
\hat{y_0} = \beta + \frac{1}{\rho^2 }   \mathbf{k}(\mathbf{x}_0)^T  K^{-1} (\mathbf{y} - \beta\mathbf{1}).
\end{align*}

\section{Generalization of CBPK and Remark \ref{rmk:sink}} \label{appendix:cbpk}
Without loss of generality, let $\beta = 0$. Expanding \eqref{eq:cbpkobj}, we get
 \begin{align*} 
& \EE[(y_0 - \lambda^T \mathbf{y})^2]+ \delta\EE[(y_0 - \EE[\lambda^T \mathbf{y}|y_0])^2]\\
 &=k(\mathbf{x}_0,\mathbf{x}_0) - 2\lambda^T \mathbf{k}(\mathbf{x}_0) + \lambda^T K \lambda + \delta \EE[(y_0 - \lambda^T \tilde{m})^2] \\
 &= k(\mathbf{x}_0,\mathbf{x}_0) - 2\lambda^T \mathbf{k}(\mathbf{x}_0) + \lambda^T K \lambda + \delta \bigg(1-\frac{\lambda^T \mathbf{k}(\mathbf{x}_0) }{k(\mathbf{x}_0,\mathbf{x}_0)} \bigg)^2k(\mathbf{x}_0,\mathbf{x}_0).
\end{align*}
This is a quadratic form of $\lambda$, and the minimizing $\lambda$ can be computed as in \eqref{eq:woodbury} by the Woodbury formula. We get
 \begin{align*} 
\hat{\lambda} &= \bigg(K+ \frac{\delta}{k(\mathbf{x}_0,\mathbf{x}_0)}\mathbf{k}(\mathbf{x}_0)\mathbf{k}(\mathbf{x}_0)^T \bigg)^{-1}(\mathbf{k}(\mathbf{x}_0)+\delta \mathbf{k}(\mathbf{x}_0)) \\
&= \frac{\delta+1}{\delta\rho^2+1} K^{-1} \mathbf{k}(\mathbf{x}_0).
\end{align*}
For $\delta \geq 0$, $w(\mathbf{x}_0) = (\delta+1)/(\delta\rho^2+1) \in [1, 1/\rho^2)$, and $\lim\limits_{\delta \rightarrow \infty} w(\mathbf{x}_0) = 1/\rho^2$. For $\delta =1/\rho$, $w(\mathbf{x}_0) = 1/\rho$ which produces the SiNK predictor. 

\section{Definition of SiNK} \label{appendix:sink}
The logarithm of the posterior probability (up to a constant) is 
\begin{align*}
\log p(y_0|\mathbf{y}) &= -\frac{1}{2}(\mathbf{y} -\tilde{m} )^T \tilde{K}^{-1} (\mathbf{y} -\tilde{m} ) - \frac{\rho}{2(1+\rho)}  \frac{(y_0 - \beta)^2}{k(\mathbf{x}_0, \mathbf{x}_0)}  
\end{align*}
Differentiating with respect to $y_0$, we get
\begin{align*}
\frac{\partial  \log p(y_0|\mathbf{y})}{\partial y_0} &= \frac{1}{k(\mathbf{x}_0,\mathbf{x}_0)}(\mathbf{y} -\tilde{m} )^T \tilde{K}^{-1} \mathbf{k}(\mathbf{x}_0)  - \frac{\rho}{1+\rho}  \frac{(y_0 - \beta)}{k(\mathbf{x}_0, \mathbf{x}_0)}  \\
&=\frac{1}{(1-\rho^2) k(\mathbf{x}_0,\mathbf{x}_0)}(\mathbf{y} -\tilde{m} )^T K^{-1} \mathbf{k}(\mathbf{x}_0)  - \frac{\rho}{1+\rho}  \frac{(y_0 - \beta)}{k(\mathbf{x}_0, \mathbf{x}_0)} 
\end{align*}
from \eqref{eq:woodbury}. Solving $\partial  \log p(y_0|\mathbf{y}) / \partial y_0= 0$ leads to 
\begin{align*}
\hat{y_0} = \beta + \frac{1}{\rho }   \mathbf{k}(\mathbf{x}_0)^T  K^{-1} (\mathbf{y} - \beta\mathbf{1}).
\end{align*}

\section{Proof of Theorem \ref{thm:localness} and Proposition \ref{property:localness}} \label{appendix:thm1}
\begin{proof}
Let the stationary variance $K(\mathbf{x}, \mathbf{x}) = \sigma^2$. Now for a target point $\mathbf{x}_0 \in B (\mathbf{x}_j) \cap J_k$, for $l \neq j$,
   \begin{align*}
\lim_{\theta_k \rightarrow 0} \frac{K(\mathbf{x}_0, \mathbf{x}_l)}{K(\mathbf{x}_0, \mathbf{x}_j)}  = \lim_{\theta_k \rightarrow 0}  \prod_{i=1}^d \frac{C_{\theta_i}(|(\mathbf{x}_l - \mathbf{x}_0)_i|)}{C_{\theta_i}(|(\mathbf{x}_j - \mathbf{x}_0)_i|)} = \lim_{\theta_k \rightarrow 0}  \prod_{i=1}^d \frac{C_1\left(\frac{|(\mathbf{x}_l - \mathbf{x}_0)_i|}{\theta_i}\right)}{C_{1}\left(\frac{|(\mathbf{x}_j - \mathbf{x}_0)_i|}{\theta_i}\right)} = 0
\end{align*}
Thus we obtain
\begin{align*}
\lim_{\theta_k \rightarrow 0} \frac{1}{K(\mathbf{x}_0, \mathbf{x}_j)} \mathbf{k}(\mathbf{x}_0) = \mathbf{e}_j
\end{align*}
where $\mathbf{e}_j$ is the $j$\nobreakdash-th unit vector. Noting that $\mathbf{x}_j \in B (\mathbf{x}_j) \cap J_k$, we have 
 \begin{align*}
\lim_{\theta_k \rightarrow 0} \frac{1}{\sigma^2}K = I_n
\end{align*}
 where $I_n$ is the $n \times n$ identity matrix. Thus,
 \begin{align*}
\lim_{\theta_k \rightarrow 0} \frac{\rho^2}{K(\mathbf{x}_0,\mathbf{x}_j)^2} &=\lim_{\theta_k \rightarrow 0} \frac{\mathbf{k}(\mathbf{x}_0)^T K^{-1}\mathbf{k}(\mathbf{x}_0)}{\sigma^2K(\mathbf{x}_0,\mathbf{x}_j)^2} = \frac{1}{\sigma^4} \;\;\mbox{and}\\
\lim_{\theta_k \rightarrow 0}\frac{\sigma^2\mathbf{k}(\mathbf{x}_0)^T K^{-1} (\mathbf{y} - \beta \mathbf{1})}{K(\mathbf{x}_0, \mathbf{x}_j)} &= y_j - \beta. 
\end{align*}
Now note that
 \begin{align*}
\hat{Y}(\mathbf{x}_0) &= \beta + w(\rho)  \mathbf{k}(\mathbf{x}_0)^T K^{-1} (\mathbf{y} - \beta \mathbf{1}) \\
&= \beta + w(\rho)  \rho \; \frac{K(\mathbf{x}_0,\mathbf{x}_j)}{\rho \sigma^2} \frac{\sigma^2\mathbf{k}(\mathbf{x}_0)^T K^{-1} (\mathbf{y} - \beta \mathbf{1})}{K(\mathbf{x}_0, \mathbf{x}_j)} .
\end{align*}
Thus, to satisfy \eqref{eq:localinlimit},
 \begin{align} \label{eq:localcond}
\lim_{\theta_k \rightarrow 0} w(\rho)  \rho =1
\end{align}
is the condition that needs to hold. For the SiNK predictor, $w(\rho) = 1/ \rho$, so the condition holds, and therefore SiNK has the localness property and Proposition \ref{property:localness} holds. 

 The limit range of $\rho$ as $\theta_k \rightarrow 0$ needs to be determined. Note that for fixed $\mathbf{x}_0 \in B(\mathbf{x}_j) \cap J_k$, $\rho \rightarrow 0$ as $\theta_k \rightarrow 0$. Now for any $\delta \in (0,1]$, let $\epsilon = C_1^{-1}(\delta)$ and $\mathbf{x}_0 = \mathbf{x}_j + \epsilon  \theta_k \mathbf{e}_k$. For all sufficiently small and positive $\theta_k$, we have $\mathbf{x}_0 \in B(\mathbf{x}_j) \cap J_k$. Then
 \begin{align*}
\lim_{\theta_k \rightarrow 0} \frac{K(\mathbf{x}_0, \mathbf{x}_j)}{\sigma^2} = \lim_{\theta_k \rightarrow 0}  \prod_{i=1}^d C_{\theta_i}((\mathbf{x}_j - \mathbf{x}_0)_i) = \lim_{\theta_k \rightarrow 0} C_{\theta_k}(\epsilon \theta_k ) = C_{1}(\epsilon) = \delta
\end{align*}
 Thus, $\lim\limits_{\theta_k \rightarrow 0}  \rho =  \delta$ for our selection of $\mathbf{x}_0$. For \eqref{eq:localcond} to hold, since $w$ is a continuous function of $\rho$, $w(\delta)  \delta = 1$ must hold for all $\delta \in (0,1]$. To put it differently, if \eqref{eq:localinlimit} holds, then it is the SiNK predictor. \qquad
\end{proof}

\section{Test Functions} \label{subsec:testfunctions}
\subsection{Borehole Function} (Morris et al. \cite{morris1993bayesian})
\begin{align*}
f(\mathbf{x}) = \frac{2\pi T_u (H_u - H_l)}{\log(r/r_w) \left(1.5 + \frac{2LT_u}{\log(r/r_w)r_w^2 K_w}+\frac{T_u}{T_l}\right)}
\end{align*}
The ranges of the eight variables are $r_w$ : (0.05, 0.15), $r$ = (100, 50000), $T_u$ = (63070, 115600), $H_u$ = (990, 1110), $T_l$ = (63.1, 116), $H_l$ = (700, 820), $L$ = (1120, 1680), and $K_w$ = (9855,12045).
\subsection{Welch} (Welch et al. \cite{welch1992screening})
\begin{align*}
f(\mathbf{x}) &= \frac{5x_{12}}{1+x_1} + 5(x_4 - x_{20})^2 + x_5 + 40x_{19}^3 - 5x_{19}\\ &+ 0.05x_2 + 0.08x_3 - 0.03x_6 + 0.03x_7 - 0.09x_9 - 0.01x_{10} - 0.07x_{11}+ 0.25x_{13}^2 \\& -0.04x_{14} + 0.06x_{15} - 0.01x_{17} - 0.03x_{18}, ~ \mathbf{x} \in [-0.5,0.5]^{20}.
\end{align*}
\subsection{Friedman} (Friedman et al. \cite{friedman1983multidimensional})
\begin{align*}
f(\mathbf{x}) = 10\sin(\pi x_1 x_2) + 20(x_3 - 0.5)^2 +10x_4 + 5x_5, ~ \mathbf{x} \in [0,1]^{5}.
\end{align*}
\subsection{Robot Arm} (An and Owen \cite{an2001quasi})
\begin{align*}
f(\mathbf{x}) &= (u^2 + v^2) ^{0.5}, \mbox{where}\\
u &= \sum\limits_{i=1}^4 L_i \cos \Big( \sum\limits_{j=1}^i \theta_j \Big),\\
v &= \sum\limits_{i=1}^4 L_i \sin \Big( \sum\limits_{j=1}^i \theta_j \Big),\\
\mathbf{x} &= (\theta_1,\ldots,\theta_4, L_1,\ldots,L_4) \in [0,2\pi]^4 \times [0,1]^4.
\end{align*}

\end{document}